\documentclass[DIV=13]{modartcl}
\pdfoutput=1

\usepackage{cleveref}
\crefname{algocount}{Algorithm}{Algorithms}
\crefname{algorithm}{Algorithm}{Algorithms}
\Crefname{algorithm}{Algorithm}{Algorithms}
\Crefname{proof}{Proof}{Proofs}
\Crefname{figure}{Figure}{Figures}
\crefformat{equation}{#2(#1)#3}
\Crefformat{equation}{Equation #2(#1)#3}
\crefmultiformat{equation}{#2(#1)#3}{ and~#2(#1)#3}{, #2(#1)#3}{ and~#2(#1)#3}
\Crefmultiformat{equation}{Equations #2(#1)#3}{ and~#2(#1)#3}{, #2(#1)#3}{ and~#2(#1)#3}
\crefrangeformat{equation}{#3(#1)#4--#5(#2)#6}
\Crefrangeformat{equation}{Equations #3(#1)#4 to #5(#2)#6}
\crefformat{enumi}{#2(#1)#3}
\crefformat{enumii}{#2(#1)#3}
\crefformat{enumiii}{#2(#1)#3}
\crefformat{enumiv}{#2(#1)#3}
\makeatletter
\def\@itemref#1#2:#3\relax{#1{#2}\;\cref{#2:#3}}
\newcommand{\itemref}[1]{\@itemref\cref#1\relax}
\newcommand{\Itemref}[1]{\@itemref\Cref#1\relax}
\makeatother
\crefname{lemma}{Lemma}{Lemmas}
\crefname{section}{Section}{Sections}
\Crefname{section}{Section}{Sections}
\crefname{table}{Table}{Tables}
\Crefname{table}{Table}{Tables}
\crefformat{appendix}{#2#1#3}

\newcommand{\ignore}[1]{}

\newcommand{\NN}{\mathbb{N}}
\newcommand{\ZZ}{\mathbb{Z}}
\newcommand{\QQ}{\mathbb{Q}}
\newcommand{\RR}{\mathbb{R}}
\newcommand{\CC}{\mathbb{C}}
\newcommand{\var}{\operatorname{var}}

\newcommand{\sep}{\ensuremath{\sigma}}

\newcommand{\tO}{\smash{\tilde{O}}}

\newtheorem{lemma}{Lemma}

\DeclarePairedDelimiter\set{\{}{\}}
\DeclarePairedDelimiter\abs{\lvert}{\rvert}

\DeclarePairedDelimiter\ceil{\lceil}{\rceil}

\definecolor{badcolor}{rgb}{0.6,0,0}
\definecolor{goodcolor}{rgb}{0,0.6,0}
\newcommand{\bad}[1]{\textcolor{badcolor}{#1}}
\newcommand{\better}[1]{\textcolor{badcolor!66.667!goodcolor}{#1}}
\newcommand{\good}[1]{\textcolor{badcolor!33.333!goodcolor}{#1}}
\newcommand{\great}[1]{\textcolor{goodcolor}{#1}}

\usepackage{array,tabularx,multirow}
\newcommand\undefcolumntype[1]{\expandafter\let\csname NC@find@#1\endcsname\relax}
\newcommand\forcenewcolumntype[1]{\undefcolumntype{#1}\newcolumntype{#1}}

\newcolumntype{L}{>{\raggedright\arraybackslash}X}
\newcolumntype{C}{>{\centering\arraybackslash}X}
\newcolumntype{R}{>{\raggedleft\arraybackslash}X}
\usepackage{booktabs}

\usepackage{multirow}
\usepackage{numprint}
\usepackage{array}
\npdecimalsign{.}
\nprounddigits{1}
\newcolumntype{H}{>{\setbox0=\hbox\bgroup}c<{\egroup}@{\!\!}|}
\forcenewcolumntype{N}{n{5}{2}}
\newcolumntype{S}{n{3}{1}}

\newcommand{\error}{\textcolor{badcolor}{\emph{err}}}
\newcommand{\timeout}{\textcolor{badcolor}{\strut\clap{$*$}\phantom{\scalebox{0.5}{.}0}\strut}}
\renewcommand{\timeout}{\textcolor{badcolor}{\strut\emph{\textgreater~600}\strut}}
\newcommand{\tablesize}{}
\newlength{\timecolwidth}
\newlength{\degcolwidth}
\newlength{\btscolwidth}
\newcolumntype{T}{>{\raggedleft\arraybackslash}p{\timecolwidth}}
\newcolumntype{D}{>{\raggedleft\arraybackslash}p{\degcolwidth}}
\newcolumntype{B}{>{\raggedleft\arraybackslash}p{\btscolwidth}}

\newcommand{\vcenterwithline}[2]{%
  \setlength{\dimen0}{.5\baselineskip * (#1 - 2)}%
  \settoheight{\dimen1}{0}%
  \settowidth{\dimen2}{#2}%
  \multirow{#1}{*}{%
      #2%
      \hspace{-0.5\dimen2}\hspace{-0.5\lightrulewidth}%
      \raisebox{\baselineskip}[0pt][0pt]{\rule{\lightrulewidth}{\dimen0}}%
      \hspace{-\lightrulewidth}%
      \raisebox{-\baselineskip + \dimen1 - \dimen0}[0pt][0pt]{\rule{\lightrulewidth}{\dimen0}}%
      \hspace{0.5\dimen2}\hspace{-0.5\lightrulewidth}%
  }%
}

\floatstyle{boxed}
\newfloat{algorithm}{htp}{algorithm}[section]
\let\oldalgorithm\algorithm
\renewcommand{\algorithm}[1][]{\oldalgorithm[#1]\raggedright}
\newcounter{algocount}
\newcommand{\algtitle}[1]{\refstepcounter{algocount}\textbf{Algorithm \ref*{algo:#1}: #1}\label{algo:#1}\smallskip}
\newcommand{\algspecialtitle}[2]{\refstepcounter{algocount}\textbf{Algorithm \ref*{algo:#2}: #1}\label{algo:#2}\smallskip}
\newcommand{\alginput}[1]{\par\noindent\hangindent=3.47em\hangafter=1\textsc{Input:} #1\smallskip}
\newcommand{\algoutput}[1]{\par\noindent\hangindent=4.47em\hangafter=1\textsc{Output:} #1\smallskip}
\newcommand{\algreports}[1]{\par\noindent\hangindent=4.47em\hangafter=1\textsc{Reports:} #1\smallskip}
\newlist{algolist}{itemize}{2}
\setlist[algolist,1]{label=\textbullet, nosep, labelindent=0em, leftmargin=1.em, after=\smallskip}
\setlist[algolist,2]{label=$\triangleright$, nosep, labelindent=0em, leftmargin=1.em}

\makeatletter

\newcommand{\flexvspace}{\vspace{1em minus 1ex}}
\def\flexvspacenoindent{%
  \flexvspace
  \par\noindent%
  \@ifnextchar\par\@gobble\relax%
}
\let\oldparagraph\paragraph
\renewcommand{\paragraph}[1]{\oldparagraph{#1}}
\makeatother

\renewcommand\ttfamily\sffamily

\newcommand{\ADsc}{\textsc{ADsc}\xspace}
\newcommand{\ANewDsc}{\textsc{ANewDsc}\xspace}

\usepackage{hyphenat}
\title{Computing Real Roots of Real Polynomials \ldots\newline\ldots\ and now For Real!}

\author[,]{1,2}{Alexander Kobel}
\author{3}{Fabrice Rouillier}
\author{1}{Michael Sagraloff}

\email{alexander.kobel@mpi-inf.mpg.de}
\email{Fabrice.Rouillier@inria.fr}
\email{michael.sagraloff@mpi-inf.mpg.de}

\affil{1}{Max-Planck-Institut für Informatik \bulletsep Campus E1\ 4, 66123 Saarbrücken, Germany}
\affil{2}{Universität des Saarlandes \bulletsep 66123 Saarbrücken, Germany}
\affil{3}{Institut National de Recherche en Informatique et en Automatique Paris \& Université Pierre et Marie Curie Paris VI \& Institut de Mathématiques de Jussieu Paris Rive Gauche \bulletsep 4 place Jussieu, 75005 Paris, France}

\hypersetup{
  pdfauthor={Alexander Kobel, Fabrice Rouillier and Michael Sagraloff},
  pdftitle={Computing Real Roots of Real Polynomials \ldots\ and now For Real!},
  pdfsubject={implementation of the asymptotically fast real root solver ANewDsc on top of RS, combining the Descartes method and Newton iteration in approximate arithmetic},
  pdfkeywords={real roots; univariate polynomials; root finding; root isolation; Newton's method; Descartes method; approximate arithmetic; certified computation}
}

\keywords{
  real roots \and
  univariate polynomials \and
  root finding \and
  root isolation \and
  Newton's method \and
  Descartes method \and
  approximate arithmetic \and
  certified computation

  \raisebox{.25\baselineskip}{\thintitlerule}

  Accepted for presentation at the \emph{41st International Symposium on Symbolic and Algebraic Computation \mbox{(ISSAC),}} July 19--22, 2016,\\Waterloo, Ontario, Canada.

  Definitive version to be published by ACM with DOI \href{https://dx.doi.org/10.1145/2930889.2930937}{10.1145/2930889.2930937}.

  \raisebox{.25\baselineskip}{\thintitlerule}

  \copyright\ April 29, 2016.
  Copyright held by the authors.
}

\begin{document}

\maketitle

\begin{abstract}
Very recent work introduces an asymptotically fast subdivision algorithm, denoted \ANewDsc, for isolating the real roots of a univariate real polynomial.
The method combines Descartes' Rule of Signs to test intervals for the existence of roots, Newton iteration to speed up convergence against clusters of roots, and approximate computation to decrease the required precision.
It achieves record bounds on the worst-case complexity for the considered problem, matching the complexity of Pan's method for computing all complex roots and improving upon the complexity of other subdivision methods by several magnitudes.

In the article at hand, we report on an implementation of \ANewDsc on top of the RS root isolator.
RS is a highly efficient realization of the classical Descartes method and currently serves as the default real root solver in Maple.
We describe crucial design changes within \ANewDsc and RS that led to a high-performance implementation without harming the theoretical complexity of the underlying algorithm.

With an excerpt of our extensive collection of benchmarks, available online at \url{http://anewdsc.mpi-inf.mpg.de/}, we illustrate that the theoretical gain in performance of \ANewDsc over other subdivision methods also transfers into practice.
These experiments also show that our new implementation outperforms both RS and mature competitors by magnitudes for notoriously hard instances with clustered roots.
For all other instances, we avoid almost any overhead by integrating additional optimizations and heuristics.

\end{abstract}

\section{Introduction}

Computing the real roots of a univariate polynomial is one of the fundamental tasks in numerics and computer algebra, and numerous methods have been proposed to solve this problem.
The leading general-purpose solvers in practice are based on subdivision algorithms that rely on Descartes' Rule of Signs to test for the existence of roots in a certain interval.

The computation for an input polynomial $P \in \RR[x]$ can be considered as a binary tree where each node corresponds to an interval $I = (a,b)$ and a polynomial $P_{I}=(x+1)^n P(\frac{ax+b}{x+1})$. The number $v_I$ of sign changes in the coefficient sequence of $P_I$ then exceeds the number of roots contained in $I$ by an even non-negative number.
In the original algorithm~\cite{Collins-Akritas}, the two children of a node are obtained by a simple bisection on the interval and relative transformations on the polynomials, which yields the polynomial~$P_I$.
A node is a leaf if $v_I=0$ or $v_I=1$; in the first case, $I$ contains no root, whereas, in the latter case, the interval is isolating for a real root.
Considerable efforts have been taken to improve the worst-case complexity of different variants of the Descartes method, and to provide efficient implementations:
different traversal orders of the subdivision tree to optimize the memory usage~\cite{Krandick:1995},
the use of interval arithmetic~\cite{Johnson-Krandick},
new strategies to optimize the main transformations~\cite{rouillier-zimmermann:roots:04},
extensions to polynomials with approximate coefficients~\cite{Mehlhorn-Sagraloff,DBLP:conf/casc/EigenwilligKKMSW05},
and many others.

Nevertheless, most formulations suffer from at least one of the following two deficiencies:
First, the subdivision strategies only achieve \emph{linear} convergence against the roots.
In presence of clusters, those approaches require a large number of subdivisions to separate the roots.
Second, the algorithms are designed for exact arithmetic on the coefficients of $P$.
An unguarded choice of the subdivision points can impose an unnecessarily large precision demand when such an approach is translated literally to arbitrary precision dyadic numbers.
Even worse, completeness is not guaranteed for polynomials whose coefficients can only be approximated.
For example, consider the polynomial $P(x) = (x-1)(\pi x-e)$ with roots at (exactly) $1$ and $e/\pi \approx 0.865$,
and assume that its coefficients are given as an oracle for arbitrarily good dyadic approximations.
In this setting, $P(1)=0$ cannot be decided without resorting to a symbolic simplification of the input, which is beyond the means of the root isolation method.
Likewise, the number of sign variations on the intervals $(1-\epsilon, 1)$ and $(1, 1+\epsilon)$ for any $\epsilon$ cannot be determined with approximate arithmetic.
Thus, once an interval is split at $x=1$, a straight-forward algorithm will request better and better approximations of $P$ from the oracle, but will not terminate.

Both shortcomings have been resolved in~\cite{Sagraloff2015}: The proposed algorithm \ANewDsc
combines the bisection strategy with Newton iteration to achieve \emph{quadratic} convergence against clusters.
Thus, long chains of intervals $I$ with the same $v_I$ are compressed to only logarithmic length (compared to the length when considering bisection only).
In addition, intervals are split at \emph{admissible points,} where the polynomial takes a (relatively) large absolute value.
This allows the precision demand of the computation to be kept small, improving by an order of magnitude over previous approaches.
Due to the exclusive use of approximate arithmetic, the method also applies to (square-free) polynomials with arbitrary real coefficients.
Both the number of subdivision steps and the precision demand is optimal up to polylogarithmic factors,
and the bound on its bit complexity is comparable to the record bound~\cite{MSW-rootfinding2013} that is implied by an algorithm based on Pan's near-optimal method~\cite{Pan02} for approximate polynomial factorization. We provide more specific bounds in \cref{sec:review}.

Unfortunately, asymptotically fast algorithms are often extremely hard to be implemented or do not exhibit their theoretical performance in practice.
The contribution of this paper is to show that, for the problem of real root computation, we can bridge this gap between theory and practice.
We report on an implementation of \ANewDsc on top of the Descartes-based real root finder inside the RS library.
It preserves the key features of RS, which are a close-to-optimal memory consumption and the intensive use of adaptive multiprecision interval arithmetic~\cite{ReRo02}.

We present our design changes both to RS and within \ANewDsc that make it possible to achieve significant performance gains on hard instances on the one hand and proven complexity guarantees on the other hand
without sacrificing efficiency for small or intrinsically easy instances.
The analysis is supported by benchmark results that show that \ANewDsc can defy the leading general-purpose solvers for real roots in their special domains,
and outperforms the existing implementations on notoriously hard instances.

\section[The Descartes Method and the ANewDsc Algorithm]{The Descartes Method and the \ANewDsc Algorithm}\label{sec:review}

We briefly review the classical Descartes method~\cite{Collins-Akritas} as well as the algorithm {\ANewDsc} as introduced in~\cite{Sagraloff2015}.
Given an interval $\mathcal{I}=(a_0,b_0)\subset\RR$ and a polynomial
\begin{align}\label{def:polyP}
  P(x)=p_nx^n + p_{n-1}x^{n-1}+\cdots+p_1x + p_0 \in\mathbb{R}[x],
\end{align}
the Descartes method recursively subdivides $\mathcal{I}$ into equally sized intervals until each subinterval $I=(a,b)\subset \mathcal{I}$ has either been shown to contain no root or exactly one root of $P$.
In order to test an interval for the existence of a root of $P$, a coordinate transformation $x\mapsto\smash{\frac{ax+b}{x+1}}$ is considered, which maps $\RR^+$ one-to-one onto the interval $I$.
Then, Descartes' Rule of Signs applied to the polynomial
\begin{align}\label{def:polyPI}
  \smash{P_I(x)\coloneqq \sum\nolimits_{i=0}^n p_{I,i}x^i\coloneqq (x+1)^{n}\cdot P\Big(\frac{ax+b}{x+1}\Big)}
\end{align}
states that the number of sign changes
$v_I \coloneqq \var(P,I) \coloneqq \var(P_I)$
in the coefficient sequence of $P_I$
exceeds the number $m_I$ of roots (counted with multiplicity) of $P$ in $I$ by a non-negative even integer. In other words,
$m_I\le v_I$ and $m_I\equiv v_I \mod 2$.
In each step of the recursion within the Descartes method, the number $v_I$ is computed. If $v_I=0$, the interval $I$ is discarded. If $v_I=1$, $I$
is stored as isolating. Intervals with $v_I>1$ are further subdivided into two equally
sized sub-intervals. In addition, it is checked whether the subdivision point, that is the midpoint $m(I)=\frac{a+b}{2}$ of $I$, is
a root of $P$.

\begin{algorithm}[hpt]
  \algtitle{Classical Descartes Method}
  \alginput{A polynomial $P$ as in \cref{def:polyP} and an interval $\mathcal{I}$.}
  \algreports{Disjoint isolating intervals
    for all
    roots of $P$ in $\mathcal{I}$.}

  \begin{algolist}
  \item Initialize the list of active intervals to $\mathcal{A}\coloneqq \{\mathcal{I}\}$.
  \item While $\mathcal{A}\neq \emptyset$:
    \begin{algolist}
    \item Remove an arbitrary $I=(a,b)$ from $\mathcal{A}$.
    \item If $v_I=0$, discard $I$ and continue.
    \item If $v_I=1$, report $I$ and continue.
    \item Otherwise, add $(a,m(I))$ and $(m(I),b)$ to $\mathcal{A}$.\newline
      \phantom{Otherwise, }If $P(m(I))=0$, additionally report $\set{m(I)}$.
    \end{algolist}
  \end{algolist}
\end{algorithm}

If the polynomial $P$ contains only simple roots in $\mathcal{I}$, the Descartes method yields isolating intervals for all these roots;
otherwise, it converges towards the roots, but does not terminate.
If $\mathcal{I}$ is chosen large enough to contain all real roots,
and all these roots are simple, the algorithm isolates all real roots of $P$.
The proof for termination relies on the well known One- and Two-Circle Theorems~\cite{eigenwillig-phd,Obrechkoff:book-english},
which provide lower bounds on the width $w(I)$ of any produced interval $I$.

Regarding the worst-case complexity of the above algorithm, we focus, only for simplicity, on the so-called \emph{benchmark problem} of computing all real roots of a square-free polynomial $P$ of degree $n$
with \emph{integer} coefficients of absolute value less than $2^\tau$.
Nevertheless, we aim to stress the fact that, for polynomials with arbitrary real coefficients, more general bounds are known~\cite{Sagraloff2014DSC,Sagraloff2015},
which are expressed in terms of the separations and the absolute values of the roots of $P$, thus being more adaptive and meaningful.

We denote by $\mathcal{T}$ the subdivision tree whose nodes are the intervals $I$ considered by the algorithm.
The function $\var(\cdot)$ is sign diminishing, that is, for any disjoint subintervals $I_1$ and $I_2$ of $I$, we have $v_{I_1} + v_{I_2} \le v_I$; e.g.~see~\cite{eigenwillig-phd} for a proof.
This implies that $\mathcal{T}$ has width at most $2n$, while the lower bound on the width of the produced intervals implies that its height is in $O(\tau+\log\sep_P^{-1})$,
where $\sep_P$ is the minimum root separation of $P$.
Since $\log\sep_P^{-1}=\tO(n\tau)$, the bound $\tO(n^2\tau)$ on the total size of~$\mathcal{T}$ follows.%
\footnote{The tilde denotes that we omit polylogarithmic factors in $n$ or $\tau$.}
A more refined argument~\cite{ESY06,eigenwillig-phd}, which takes into account the fact that $\sum_{i=1}^n\log\sep_P(z_i)^{-1}=\tO(n\tau)$, even shows that $|\mathcal{T}|=\tO(n\tau)$.
The coefficients of the polynomials $P_I$ have bitsize bounded by $\tO(n^2\tau)$, hence the precision demand for exactly computing $P_I$ is also bounded by $\tO(n^2\tau)$.
This yields the bound $\tO(n^3\tau)$ for the bit complexity of computing $P_I$, even when using asymptotically fast methods for polynomial arithmetic.
We conclude that the cost of the classical Descartes method is bounded by $\tO(n^4\tau^2)$ bit operations.
This matches the bound achieved by many other popular subdivision algorithms for real root computation
such as the continued fraction (CF) method~\cite{Collins,DBLP:journals/tcs/Tsigaridas13}, the Bolzano method~\cite{Yap-Sagraloff}, or the Sturm method~\cite{Du2007}.
We remark that the $\tO(n^4\tau^2)$ upper bound on the bit complexity is actually tight for the latter algorithms; see~\cite{ESY06,Collins,Du2007}.

\flexvspacenoindent

The main strengths of the Descartes method are its simplicity and the fact that tree subdivision tree size and running time adapt to the separation of the roots~\cite{ESY06,eigenwillig-phd,Sagraloff2014DSC}.
On well-conditioned inputs, the performance is significantly better than indicated by the worst-case bounds. Indeed, we observe that, for many polynomials (e.g. random polynomials), the size of the tree is only logarithmic in $n$, in which case the precision demand does not exceed $\tO(\tau+n)$. In such cases, the bit complexity is by several magnitudes smaller than the worst case bound. However, despite its good behavior for many instances, the method has two major shortcomings.

First, in order to determine the sign of the coefficients of $P_I$ as well as the sign of~$P$ at the midpoint of $I$, exact arithmetic is assumed.
If the coefficients of $P$ are integer or rational, this is feasible, at the expensive of a precision that exceeds the actual demand by one order of magnitude; see \cref{table:worst-case-complexities} and~\cite{Sagraloff2014DSC,Sagraloff2015} for more details.
However, if the coefficients of $P$ can only be approximated (e.g~ trigonometric expressions), then computing the sign of the coefficients of $P_I$ is impossible in general.
Another shortcoming is that the classical subdivision strategies only achieve linear convergence against the roots.
For some inputs, this is not critical as the separations of the roots are large and the subdivision tree has small height.
However, if roots appear in clusters, numerous subdivision steps have to be performed in order to separate distinct roots from each other, and
there exist polynomials\footnote{As an example, consider the Mignotte-like polynomials from \cref{table:mignotte} having two real roots with pairwise distance $2^{-\Omega(n\tau)}$.}
for which the Descartes method produces a sequence of intervals $I_j$ of length $\Omega(n\tau)$ with identical $v_{I_j}$.

\begin{table}[htb]
  \centering
  \begin{tabular*}{0.8\linewidth}{l@{\extracolsep{\fill}}ccc}
    \toprule
    subdivision rule and                                                    & subdivision        & precision                              & overall bit                  \\
    model of computation                                                    & tree size          & demand                                 & complexity                   \\
    \midrule
    \bad{bisection,}                           \bad{exact over $\ZZ$/$\QQ$} & \bad{$\tO(n\tau)$} & \bad{$\tO(n^2\tau)$}                   & \bad{$\tO(n^4\tau^2)$}       \\
    \bad{bisection,}                           \great{approximate}          & \bad{$\tO(n\tau)$} & \great{$\tO(n\tau)$}                   & \better{$\tO(n^3\tau^2)$}    \\
    \rlap{\great{Newton,}}\phantom{bisection,} \bad{exact over $\ZZ$/$\QQ$} & \great{$\tO(n)$}   & \bad{$\tO(n^2\tau)$}$\mathrlap{^\ast}$ & \good{$\tO(n^3\tau)$}        \\
    \rlap{\great{Newton,}}\phantom{bisection,} \great{approximate}          & \great{$\tO(n)$}   & \great{$\tO(n\tau)$}$\mathrlap{^\dag}$ & \great{$\tO(n^3 + n^2\tau)$} \\
    \bottomrule
    \addlinespace
    \multicolumn{4}{l}{\footnotesize All mentioned bounds are tight for certain classes of inputs.}                                                                                    \\
    \multicolumn{4}{l}{\footnotesize $^\ast$Amortized precision demand over the entire tree is only \rlap{$\tO(n\tau)$.}}                                                              \\
    \multicolumn{4}{l}{\footnotesize $^\dag$Amortized precision demand over the entire tree is only \rlap{$\tO(n+\tau)$.}}
  \end{tabular*}
  \caption{worst-case complexity of variants of Descartes methods for polynomials in $\ZZ[x]$ with degree $n$ and bitsize $\tau$}
  \label{table:worst-case-complexities}
\end{table}

In~\cite{Sagraloff2015}, the algorithm {\ANewDsc}\footnote{The algorithm is an approximate arithmetic variant of the algorithm \textsc{NewDsc} from~\cite{NewDsc}, which exclusively uses exact arithmetic.
  The acronym {\ANewDsc} should be read as ``A New Descartes'' or, alternatively, as ``Approximate Arithmetic Newton-Descartes.''}
has been introduced. It uses approximate arithmetic and an accelerated subdivision strategy based on Newton's iteration to address both of the aforementioned shortcomings.
We sketch the main ideas to an extent as needed in the discussion of its implementation.
At some points, we simplified the presentation at the cost of mathematical rigor by skipping some rather technical details.
In contrast, our implementation takes into account all details, thus being certified and complete.

{\ANewDsc} is similar to the classical Descartes method in the sense that it recursively subdivides a given interval $\mathcal{I}$ and that it uses a predicate based on Descartes' Rule of Signs to test for roots.
However, it is tailored to rely solely on approximate computation.
For this, \ANewDsc uses the so called ``01-Test'',\footnote{The 01-Test is split into two separate subroutines.
  Both are modified variants of the classical sign variation tests, where it is checked whether $\var(P,J)=0$ or $\var(P,J)=1$ for some intervals~$J$. For details and proofs regarding the 01-Test, we refer the reader to~\cite[Sec.~2.4]{Sagraloff2015}.}
that can be evaluated using approximate arithmetic only.
Similar to the original Descartes test, a call of the predicate on an interval $I$ returns a value $t_I \in \set{0,1,\ast}$ with the following guarantees:
If $t_I = 0$, then $I$ contains no root of $P$; if $t_I = 1$, then $I$ isolates a simple root of $P$; and
if $t_I = {\ast}$, no conclusion shall be drawn about the number of roots of $P$ in $I$.
We remark that the 01-Test is stronger than the classical sign variation test in the sense that $v_I = 0$ (or~$1$) implies that the 01-Test returns $0$ (or~$1$), too.

The precision demand of the 01-Test on $(a,b)$ is bounded by
\begin{align*}
  \tO(n+n\log |a|+n\log |b|+\log |P(a)|^{-1}+\log |P(b)|^{-1}+\log \|P\|_\infty))
\end{align*}
bits.\footnote{We remark that the precise bound on the precision demand is slightly more complicated; see~\cite[Sec.~2.4]{Sagraloff2015} for details.}
Hence, we strive to choose boundary points $a$ and $b$ where the value of $\abs{P}$ is not too small.
This is realized by an additional layer on the subdivision scheme.
Instead of choosing a fixed subdivision point $m$ (in the classical method, the midpoint $m(I)$ of an interval $I$) in each step, \ANewDsc chooses a nearby, so-called \emph{admissible point} $m^*$, where $\abs{P}$ becomes large.
More precisely, for a point $m$, a positive integer $N$ with $N\ge n$, and a positive value $\epsilon$, we call a point
\begin{align*}
  m^*\in m[\epsilon;N]\coloneqq \set{m_i\coloneqq m+i\cdot\epsilon \text{ for } i=-\ceil{\tfrac{N}{2}},\ldots,\ceil{\tfrac{N}{2}}}
\end{align*}
\emph{admissible with respect to the multipoint $m[\epsilon;N]$ (or just admissible)} if $|P(m^*)|\ge \smash{\frac{1}{4}}\cdot \max_i |P(m_i)|$.
\cref{alg:find-admissible-point} computes such a point using only approximate arithmetic.%
\footnote{Note that in~\cite[p.~54]{Sagraloff2015}, there is a typo in the Admissible Point algorithm: \emph{repeat …until} (loop (2)) should read \emph{repeat …while}.
  However, the proof of~\cite[Thm.~5 (b)]{Sagraloff2015} uses the correct statement.}
Indeed, since $m[\epsilon;N]$ contains at least one point that has distance $\epsilon/2$ or more to all roots of $P$, we have $\max_i |P(m_i)|>0$, and thus the algorithm terminates with a precision $\rho$ that is bounded by $2\cdot (1+\log \max(1, (\max_i |P(m_i)|)^{-1}))$.

\begin{algorithm}[htb]
  \algtitle{Find Admissible Point}
  \alginput{Polynomial $P$ as in \cref{def:polyP} and
    a multipoint $m[\epsilon; N]$.}
  \algoutput{Admissible point $m^* \in m[\epsilon; N]$.}

  \begin{algolist}
  \item For $\rho = 2, 4, 8, \dots$:
    \begin{algolist}
    \item For all $m_i\in m[\epsilon;N]$, compute approximations $\tilde{v}_i$
      of $v_i \coloneqq P(m_i)$ with $|\tilde{v}_i-v_i|<2^{-\rho}$.
    \item Determine a point $m_{i_0}=m_i$
      that maximizes $|\tilde{v}_{i}|$.
    \item If $\abs{\tilde v_{i_0}}>2^{-\rho+2}$, return $m^*\coloneqq m_{i_0}$.
    \end{algolist}
  \end{algolist}
  \label{alg:find-admissible-point}
\end{algorithm}

In~\cite{Sagraloff2015}, we always choose $N=n$. Using approximate multipoint evaluation~\cite{DBLP:journals/corr/KobelS14}, this guarantees that the cost for the approximate evaluations of $P$ at all points $m_i$ does not considerably exceed the cost for the evaluation at only one point. In addition, in each step, $\epsilon$ is chosen small enough (e.g., $\epsilon\approx w(I)/4$ in a bisection step with $m = m(I)$) such that the size and the structure of the subdivision tree induced by {\ANewDsc} does not change (in the worst case)
when passing from an admissible point $m^*\in m[\epsilon;N]$ to an arbitrary point in the interval $[m-\ceil{\tfrac{N}{2}}\epsilon,m+\ceil{\tfrac{N}{2}}\epsilon]$. In fact, choosing $m^*$ to be admissible only affects the needed precision demand in each iteration but not the size of the subdivision tree.

For simplicity, we will not further specify $\epsilon$ (nor $N$) throughout the following considerations and just assume that $\epsilon$ is chosen small enough; for details, see~\cite{Sagraloff2015}.
In this context, we just say that a point $m^*$ is \emph{admissible for $m$}.
With a suitable choice of the parameters, subdivision on admissible points reduces the worst-case precision demand of the Descartes method to $\tO(n\tau)$ in the integer setting,
improving upon the classical approach by one order of magnitude.
Besides, it allows to process inputs where the coefficients can only be approximated as we avoid splitting on roots \emph{without} verifying that $P$ is zero at any point.

\flexvspacenoindent
In order to overcome the second shortcoming of the Descartes method, {\ANewDsc} combines bisection with
Newton iteration. For this, it uses a trial and error approach, called ``Newton-Test'', to speed up
convergence towards clusters of roots. In each iteration,
the Newton-Test aims to replace an interval $I=(a,b)$ by some sub-interval $I'=(a',b')\subset I$ of width
$w(I')\approx w(I) / N_I$ such that $I'$ contains all roots of $P$ in $I$.
Here, $N_I$ denotes a parameter that corresponds to the actual speed of convergence.
Initially, $N_I$ is set to $4$.
If the Newton step succeeds, we define $N_{I'}\coloneqq N_I^2$.
In case of failure, we fall back to bisection, that is, $I$ is subdivided into two (almost) equally-sized subintervals $I_{\ell}$ and $I_r$, and we define $N_{I_{\ell}} \coloneqq N_{I_r}\coloneqq \max(4,\sqrt{N_I})$.
In~\cite[Sec.~3.2]{Sagraloff2015}, it has been shown that the Newton step succeeds under guarantee if there exists a sufficiently small cluster $\mathcal{C}$ of roots of $P$
that is centered at some point in $I$ and is sufficiently well-separated from the remaining roots of $P$.

\begin{algorithm}[ht]
  \algtitle{Newton-Test}
  \alginput{A polynomial $P$ as in \cref{def:polyP}, an interval $I=(a,b)$, and\newline an $N_I\in\NN$ of the form $N_I=\cramped{2^{2^l}}$ with $l\in\NN_{\ge 1}$.}
  \algoutput{\textsc{false} or an interval $I'\subset I$, with $\frac{N_Iw(I')}{w(I)}\in [\frac{1}{8},1]$, that contains all roots of $P$ in $I$.}

  \begin{algolist}
  \item For $j \in \set{1,2,3}$, let $\xi_{j}\coloneqq a+\frac{j}{4}\cdot w(I)$,
    compute admissible points $\xi_{j}^{*}$ for $\xi_{j}$ and\\
    the Newton correction terms $v_{j}\coloneqq P(\xi_{j}^{*}) / P'(\xi_{j}^{*})$.
  \item For all pairs $i,j \in \set{1,2,3}$ with $i < j$:
    \begin{algolist}
    \item Compute approximations $\tilde{\lambda}_{i,j}$ of
      $\lambda_{i,j}\coloneqq \xi_{i}^{*}+\tilde{k} \cdot v_i$, where $\tilde{k}\coloneqq (\xi^{*}_{j}-\xi^{*}_{i})/(v_{i}-v_{j})$,\\
      with $\abs{\tilde{\lambda}_{i,j}-\lambda_{i,j}} \le 1/32N_{I}$.
    \item If $\tilde{\lambda}_{i,j} \not\in[a,b]$, discard $(i,j)$.
      Otherwise, define\\
      \quad $\ell_{i,j}\coloneqq \lfloor (\tilde{\lambda}_{i,j}-a) \cdot 4N_I/w(I) \rfloor\in\{0,\ldots,4N_{I}\}$,\\
      \quad $a_{i,j}\coloneqq a+\max\set{0,\ell_{i,j}-1}\cdot w(I) / 4N_{I}$, and\\
      \quad $b_{i,j}\coloneqq a+\min\set{4N_{I},\ell_{i,j}+2}\cdot w(I) / 4N_{I}$.\\
      If $a_{i,j}=a$, set $a_{i,j}^{*}\coloneqq a$, and if $b_{i,j}=b$, set $b_{i,j}^*\coloneqq b$.\\
      Otherwise, compute admissible points $a_{i,j}^{*}$ and $b_{i,j}^{*}$
      for $a_{i,j}$ and $b_{i,j}$.
    \item If the 01-Test returns $0$ for both intervals $(a,a_{i,j}^*)$
      and $(b_{i,j}^*,b)$, return $I' = (a_{i,j}^{*},b_{i,j}^{*})$.\\
      Otherwise, discard the pair $(i,j)$.
    \end{algolist}
  \item \emph{Boundary test:}
    If all pairs have been discarded,\\
    compute admissible points $m_{\ell}^*$ and $m_r^*$ for $m_\ell\coloneqq a+w(I)/2N_I$ and $m_r\coloneqq b-w(I)/2N_I$.
    \begin{algolist}
    \item If the 01-Test yields $0$ for $(m_{\ell}^*,b)$, return $(a,m_{\ell}^*)$.
    \item If the 01-Test yields $0$ for $(a,m_{r}^*)$, return $(m_{r}^*,b)$.
    \item Otherwise, return \textsc{false}.
    \end{algolist}
  \end{algolist}
  \label{alg:newton-test}

  \medskip
  \footnoterule\smallskip
  {\small\noindent
    This is a simplified version of Newton- and Boundary-Test from~\cite{Sagraloff2015}.
    For the sake of a concise presentation, some technical details (e.g.~concerning the precision management) are omitted.}
\end{algorithm}

\newpage
The method runs in three steps.
In the first step, it estimates the multiplicity $k$ of such a cluster.%
\footnote{\Cref{alg:newton-test} only computes some rational value $\smash{\tilde{k}}$.
  Under the assumption that a cluster $\mathcal{C}$ as above exists, we have $k\approx\smash{\tilde{k}}$ and, in particular, $k$ equals the integer that is closest to $\smash{\tilde{k}}$.
}
Then, in the second step, it performs a corresponding Newton iteration to obtain a $\lambda$ with $|\lambda-z_i|<w(I) / 2N_I$ for each root $z_i\in\mathcal{C}$.
Finally, the method aims to validate the approximation $\lambda$ of the cluster $\mathcal{C}$.
For this, we apply the 01-Test to the intervals $(a,a')$ and $(b,b')$, where $I'=(a',b')\subset I$ is an interval of width $w(I')\approx w(I) / N_I$ centered at $\lambda$.
If the latter test yields the value $0$ for both intervals $(a,a')$ and $(b,b')$, it follows that $I'$ contains all roots that are contained in~$I$, in which case we say that the Newton-Test succeeds.
Notice that, even in case of success, we may not conclude that there exists a cluster of $k$ roots, however, we may conclude that $I$ may be replaced by $I'$ without discarding any root.

It should further be mentioned that the Newton-Test also integrates a so-called Boundary-Test, which, by default,
always checks whether one of the two intervals $(a,a+w(I)/N_I)$ or $(b-w(I)/N_I,b)$ contains
all roots that are contained in~$I$. Again, for this, the 01-Test is applied to $(a+w(I)/N_I,b)$ and $(a,b-w(I)/N_I)$, respectively; see \cref{algo:ANewDsc} and~\cite[Sec.~3.2]{Sagraloff2015} for more details.

\begin{algorithm}[ht]
  \algspecialtitle{\ANewDsc}{ANewDsc}
  \alginput{A polynomial $P$ as in \cref{def:polyP} and an interval $\mathcal{I}$.}
  \algreports{Disjoint isolating intervals
    for all
    roots of $P$ in $\mathcal{I}$.}

  \begin{algolist}
  \item Initialize the list of active intervals to $\mathcal{A}\coloneqq \{(\mathcal{I},4)\}$.
  \item While $\mathcal{A}\neq \emptyset$:
    \begin{algolist}
    \item Remove an arbitrary $(I, N_I)$ with $I = (a,b)$ from $\mathcal{A}$.
    \item If the 01-Test on $I$ returns 0,
      discard $I$ and continue.
    \item If the 01-Test on $I$ returns 1,
      report $I$ and continue.
    \item Otherwise:\quad\emph{Quadratic (or Newton) step}\\
      If the Newton-Test on $P$ and $(I,N_I)$ returns an interval $I'$, add $(I', N_I^2)$ to $\mathcal{A}$ and continue.
    \item Otherwise:\quad\emph{Linear (or bisection) step}\\
      Compute an admissible point $m^\ast$ for $m(I)$ and
      add $(I',N_{I'})$ and $(I'',N_{I''})$ to $\mathcal{A}$,\\where $I'=(a,m^*)$, $I''=(m^*,b)$, and $N_{I'}\coloneqq N_{I''}\coloneqq \max(4,\sqrt{N_{I}})$.
    \end{algolist}
  \end{algolist}
\end{algorithm}

In~\cite{Sagraloff2015}, the bit complexity bit complexity as well as the bounds on the number of iterations and the precision demand of {\ANewDsc} are stated in terms of the
degree of $P$ and values that exclusively depend on the roots $z_i$ of the polynomial, such as the
product of the absolute values of all roots beyond the unit disk or the product of the pairwise distances between any two roots.
Those bounds are both more adaptive to the intrinsic complexity of the problem and apply for inputs with approximable coefficients.
For brevity, we only restrict to the benchmark problem here:
The length of a sequence of intervals $I_j$ with invariant $\var(P,I_j)$
is now upper bounded by $\smash{\tO(\log(n\tau))}$ compared to $\smash{\tO(n\tau)}$ for the classical Descartes method.
This is due to the fact that the Newton-Test succeeds and achieves quadratic convergence for all but $\tO(\log(n\tau))$ intervals.
As a consequence, the overall size of the subdivision tree is bounded by $O(v_{\mathcal{I}}\cdot \log(n\tau))=O(n\log(n\tau))$, which is near-optimal; see also~\cite[Theorems~27--29]{Sagraloff2015}.
The bit complexity of the algorithm is bounded by $\tO(n^3+n^2\tau)$,
which is by three magnitudes better than the bound for the Descartes method (also see \cref{table:worst-case-complexities}) and comparable to the best bound known~\cite{MSW-rootfinding2013,Pan:survey}
for the benchmark problem as achieved by an algorithm based on Pan's near-optimal method~\cite{Pan02} for approximate polynomial factorization.
The bound is obtained by an amortized analysis of the cost at each node in the subdivision tree.
The cost for processing a certain interval $I$ is related to the maximum of all products $\max_{z_i} \sigma_P(z_i)\cdot |P'(z_i)|$, where the $z_i$'s range over all roots in the one-circle region of $I$.
In the worst-case, the cost at a node is of size $\tO(n^2\tau)$, whereas, in average, it is of size $\tO(n(n+\tau))$.
The precision demand is upper bounded by $\tO(n\tau)$ in the worst case and of size $\tO(n+\tau)$ in average.

\section[Implementing ANewDsc inside RS]{Implementing \ANewDsc inside RS}
\label{sec:implementation-details}

The RS library contains a generic framework~\cite{rouillier-zimmermann:roots:04} that allows to instantiate several variants of the Descartes method.
The choices include several bisection variants as well as continued fraction-based subdivision.
The default configuration, which we will denote as ``RS'' throughout the following considerations, achieves its high efficiency
by a careful low-level implementation of the bottleneck subroutines as well as a few sophisticated and---at that time---innovative design choices.

RS uses a hybrid arithmetic strategy.
Starting with a low precision (of 63 bits), all predicates are evaluated in interval arithmetic using MPFI~\cite{ReRo02}.
Whenever a predicate cannot be conclusively decided within the working precision, the computation is interrupted and later resumed with a higher precision.
When it is conceivable that exact arithmetic is less costly than interval arithmetic at high precision, a heuristic can trigger exact computation;
this facilitates further optimizations like (virtual) deflation of exact dyadic roots.

The realization of the sign variation test in RS is adapted to this interval setting:
While we work with exact polynomials $f = \sum f_i x^i$ in theory, the computations are executed on coefficient-wise interval approximations $[f] = \sum [f]_i x^i$ of $f$ such that $f_i \in [f]_i$ for all $i$.
In particular, the value of $\var(f)$ is replaced by a (super-)set $\var([f]) = \set{v^-, \dots, v^+}$ of possible sign variations for the family of polynomials in $[f]$, that is, $\var(g) \in \var([f])$ for all $g \in [f]$.
In this model, a sign variation test for a polynomial $[f]$ succeeds if and only if $\var([f]) = \set{0}$, $\var([f]) = \set{1}$, or $0,1 \notin \var([f])$.
Otherwise, the accuracy is not sufficient to decide the branching strategy for the current interval, and the precision is increased.
We keep this classical test as a filter in \ANewDsc; if it does not succeed, we call the full 01-Test, which again uses the interval sign variation computation as a primitive.

Another crucial difference between RS and the previous variants of the Descartes method is the traversal of the subdivision tree.
In Collins' and Akritas'~\cite{Collins-Akritas} formulation, the subdivision tree is explored in a depth-first search strategy.
This traversal allows to obtain some of the intermediate polynomials with comparatively little computational effort from previous results, but comes at the expense of storing a potentially very long list of active nodes.
On the other hand, Krandick's variant~\cite{Johnson-Krandick} with breadth-first traversal is more memory-efficient, but requires more expensive arithmetic operations; see~\cite[Sec.~3]{rouillier-zimmermann:roots:04}.
The trade-off which proved optimal for RS was a depth-first search traversal in a near-constant-memory variant, oblivious of intermediate results.
To keep the arithmetic cost of the algorithm close to the optimum for bisection methods, RS uses specifically tailored subroutines for reconstructing the discarded intermediate polynomials.
They are supplemented by a custom garbage collector that ensures that only an insignificant amount of time is spent for memory management.
For details of those design choices, we refer to the original description of RS~\cite[Sec.~4]{rouillier-zimmermann:roots:04}.

\flexvspacenoindent

RS is a promising basis for an implementation of \ANewDsc due to its high performance as a general-purpose solver and the availability in Maple.
Also, the extensive use of interval arithmetic offers the flexibility to design a numerical root solver for polynomials with irrational coefficients that can be approximated arbitrarily well.
However, to combine the theoretical improvements of the new algorithm with the long-evolved insights for a practical realization, modifications are required on both ends.
In this section, we will describe the most crucial design decisions in this light.

\subsection[Changes over the theoretical ANewDsc]{Changes over the theoretical \ANewDsc}
\label{sec:changes-to-anewdsc}

\paragraph{Random sampling of pseudo-admissible points}
\addcontentsline{toc}{subsubsection}{Random sampling of pseudo-admissible points}\label{par:pseudo-admissible}
A tight analysis of the theoretical worst-case bit complexity of any variant of the Descartes method heavily relies on asymptotically fast algorithms for polynomial arithmetic.
However, in practice, many asymptotically fast methods are only effective for inputs larger than a threshold that is well beyond what can realistically be handled with state-of-the-art solvers.
Until now, careful implementations of the naive algorithms are more efficient except for artificial counterexamples.

The most prominent places where such tools are used in \ANewDsc are the basis transformations from \cref{def:polyP} to \cref{def:polyPI} (so-called \emph{Taylor shifts})
and the approximate multipoint evaluations within the admissible point selections.
With the advanced
methods, both subroutines require only a near-linear number of arithmetic operations in the coefficient ring, compared to quadratically many for naive approaches.
The Taylor shifts are intrinsic to
any Descartes method (although we will discuss a partial remedy for some situations at the end of \cref{sec:changes-to-rs}).
But the multipoint evaluations in the ``Admissible Point'' routine can be avoided:
In our initial description of \ANewDsc, we choose an admissible point $m^*\in \mathbf{m}:=m[\epsilon;n]$ such that $\abs{P(m^*)} \ge \frac{1}{4} \cdot \max_{m_i\in\mathbf{m}} \abs{P(m_i)}$. For our implementation, we propose to consider the multipoint $\mathbf{M}:=m[\frac{\epsilon}{2^{\lambda}};n\cdot 2^{\lambda}]$ instead, which spans the same range as $\mathbf{m}$ (i.e. the extremal points in $\mathbf{m}$ and $\mathbf{M}$ define the same interval
 $
 [m-\ceil{\tfrac{n}{2}}\epsilon,m+\ceil{\tfrac{n}{2}}\epsilon]$) but contains $2^{\lambda}$ times as many points as $\mathbf{m}$. Here, $\lambda$ is a positive integer of size $O(\log n)$. However, instead of choosing an admissible point in $\mathbf{M}$ we choose a so-called \emph{pseudo-admissible point}, where $P$ is only required to take an absolute value that is related to the current working precision.
 For this, we randomly sample a point $m_j$ from $\mathbf{M}$ and compute an approximation $\tilde{v}_j$ of $v_j=P(m_j)$ (via interval arithmetic) with $|\tilde{v}_j-v_j|<2^{-\rho}$. If $|\tilde{v}_j|>2^{-\rho+2}$, we keep $m_j$. Otherwise, we double the precision $\rho$ and choose a new sample point; see also Algorithm~\ref{alg:find-pseudo-admissible-point}.

\begin{algorithm}[ht]
  \algtitle{Find Pseudo-Admissible Point}
  \alginput{A polynomial $P(x)$ as in \cref{def:polyP} and a multipoint\newline
    $\mathbf{M}:=m[\frac{\epsilon}{2^{\lambda}}; 2^{\lambda} n]$ with $\lambda\in \mathbb{N}_{\ge 2}$ and $\lambda=O(\log n)$.}
  \algoutput{A point $m' \in \mathbf{M}$ with value $P(m') \ne 0$.}

  \begin{algolist}
  \item For $\rho = 63, 127, 255, \dots$:
    \begin{algolist}
    \item Pick a random $m_j$ among the points in $\mathbf{M}$.
    \item Compute an approximation $\tilde{v}_j$ of $v_j = P(m_j)$
      with $|\tilde{v}_j-v_j|<2^{-\rho}$\newline (using interval arithmetic).
    \item If $\abs{\tilde{v}'}>2^{-\rho+2}$, return $m' \coloneqq m_j$.
    \end{algolist}
  \end{algolist}
  \label{alg:find-pseudo-admissible-point}
\end{algorithm}

The following lemma guarantees that, with high probability, we choose a point $m'\in \mathbf{M}$ for which $|P(m')|$ is not much smaller than  $\max_{m_i\in\mathbf{m}} \abs{P(m_i)}$.

\begin{samepage}
  \begin{lemma}\label{lem:expectation-pseudo-admissible-point}
    Algorithm~\ref{alg:find-pseudo-admissible-point} returns a point $m'\in \mathbf{M}$ such that
    \begin{align*}
      \smash{\log |P(m')|= 2^k\cdot O(n\log n+\log\max(1, (\max_{m_i\in\mathbf{m}} \abs{P(m_i)})^{-1}))}
    \end{align*}
    with probability $1-2^{-k\lambda}\ge 1-2^{-k}$.
  \end{lemma}
\end{samepage}

\begin{proof}
  \newcommand{\mmax}{\ensuremath{\bar{m}}}
  Let $\mathbf{M}^*\subset\mathbf{M}$ denote the subset consisting of all points in $\mathbf{M}$ whose distance to all roots of $P$ is at least $\epsilon\cdot \cramped{2^{-\lambda}}$.
  Let $\mmax\in\mathbf{m}$ be a point with $|P(\mmax)|=\cramped{\max_{m_i\in\mathbf{m}} \abs{P(m_i)}}$, and let $\cramped{\rho_0 \coloneqq \lceil\log \max(1,|P(\mmax)|^{-1})\rceil}$.
  For any complex root $z$ of $P$ and any point $m_j\in\mathbf{M}^*$, we have $\cramped{\frac{|\mmax-z|}{|m_j-z|}}>\cramped{\frac{1}{1+n2^{\lambda}}}>\cramped{2^{-\lambda-1-\log n}}$ and, thus, $\cramped{\frac{|P(m_j)|}{|P(\mmax)|}}>\cramped{2^{-n(\lambda+1+\log n)}}$.
  Suppose that, in a certain iteration of Algorithm~\ref{alg:find-pseudo-admissible-point}, we have $\rho\ge \cramped{\rho_0} + n(\lambda+1+\log n) + 2$.
  Then, it holds that $\cramped{2^{-\rho+2}}<\cramped{|P(m_j)|}$ for all $m_j\in\mathbf{M}^*$; hence, the algorithm terminates if we pick such an $m_j$.
  Since $\mathbf{M}^*$ contains at least $\cramped{2^{\lambda}(n+1)}-2n$ points, the probability of choosing a point from $\mathbf{M}^*$ is at least $1-\cramped{2^{-(\lambda-1)}}$.
  Hence, with probability $\cramped{1-2^{-k(\lambda-1)}}$, we terminate with a $\rho$ of size $\cramped{2^k}\cdot O(n\log n+\log \max(1,(\max_{m_i\in\mathbf{m}} \abs{P(m_i)})^{-1}))$.
\end{proof}

Our analysis shows that, in expectation,
there might be an increase in precision by an additive term of size $O(n\log n)$ when choosing a pseudo-admissible point instead of an admissible point;
in practice, we observe that the increase in precision is entirely negligible.
We further remark that the complexity bounds as derived in~\cite{Sagraloff2015} are not affected, since a term of size $\smash{\tO(n)}$ already appears in the bound on the precision demand of the 01-Test.
Yet, our implementation only models the theoretical result in a Las-Vegas setting due to the random choice of a pseudo-admissible point.

\paragraph{Heuristics to delay invocations of the Newton step}
\addcontentsline{toc}{subsubsection}{Heuristics to delay invocations of the Newton step}\label{par:delayed-newton}
The asymptotic cost of the calls to the Newton-Test is dominated by the Taylor shift which is part of the sign-variation test for any interval.
However, in practice there is a noticeable impact of unsuccessful Newton-Tests on the performance due to the expensive 01-Tests for the siblings $(a, a^*_{i,j})$ and $(b^*_{i,j}, b)$ of the candidate intervals $(a^*_{i,j}, b^*_{i,j})$.
Thus, we strive for calling the Newton-Test only if it will succeed with high likelihood;
vice versa, in the generic situation of well-distributed roots without any distinguished root clusters, as little computation time as possible should be wasted.

We say that a linear step on an interval $I$ to $I'$ and $I''$ leads to a \emph{proper split} if $\var(P, I')$ and $\var(P, I'')$ are both non-zero.
In this case, Obreshkoff's One-circle theorem~\cite{eigenwillig-phd,Obrechkoff:book-english} asserts that there is at least one root in the neighborhood of $I''$ that is not part of a potential cluster within $I'$, and the symmetric argument holds for $I''$.
There are two explanations for such a situation:
If the goal on the convergence speed was too optimistic (that is, the value of $N_I$ was set too high), the candidate interval in the Newton-Test might not comprise the entire cluster of roots.
In this case, a success of the Newton step on $I'$ and $I''$ with a coarser resolution $N_{I'} = N_{I''} = \sqrt{N_I}$ is not entirely unlikely.
However, if $N_I = 4$ was already at the minimum, the roots near $I'$ and $I''$ are well-separated compared to the resolution $N_I$, and an immediate success of the Newton-Test would be an unexpected artifact.

We installed a heuristic to distinguish those situations and inhibit Newton-Tests in the latter case.
For each node $I$ in the subdivision tree, we keep track of the distance to its closest ancestor $J$ with $N_J = 4$ that lead to a proper split.
If this distance is smaller than $\log n$, we immediately process $I$ with bisection%
\footnote{In the terminology of~\cite{Sagraloff2015}, \emph{proper splits} occur on a subset of the \emph{special nodes} of the subdivision tree,
  and we allow paths of $\ceil{\log n}$ many \emph{ordinary nodes} before further Newton-Tests.}
and keep $N_I = 4$.
On intervals with well-separated roots where we never achieve nor profit from quadratic convergence, this strategy renders it likely that there is never even an attempt to call a Newton-Test,
and the overhead of \ANewDsc compared to the variant \ADsc without the Newton-Test is almost zero.
On the other hand, the size of the subdivision tree is increased by at most a factor of size $O(\log n)$, and thus stays soft-linear in $n$ and polylogarithmic in $\tau$ in the worst case.

\subsection{Changes over the classical RS}
\label{sec:changes-to-rs}

\paragraph{Full caching instead of constant-memory version}
\addcontentsline{toc}{subsubsection}{Full caching instead of constant-memory version}\label{par:full-caching}
The default instantiations of RS use a near-constant-memory strategy.
None of the intermediate local polynomials are cached, but are computed from the previously considered node.
In general, this means that larger round-off errors are accumulated, and expensive recomputation from scratch is potentially required more often.
Yet, performance comparisons show a clear benefit of this approach due to the drastically reduced costs for memory management
and RS' sophisticated way of performing incremental computations between adjacent nodes in the subdivision tree~\cite[Sec.~4]{rouillier-zimmermann:roots:04}.

However, the fast transformations rely on certain structural properties of the subdivision tree:
in a pure bisection algorithm, all emerging intervals are of the form $(\frac{c}{2^e}, \frac{c+1}{2^e})$ for some integers $c$ and $e$.
In this setting and with the strong specification of the traversal order of the tree, a good portion of the arithmetic operations for Taylor shifts can be achieved by cheap bitshift and addition operations.
In contrast, a general Taylor shift requires noticeably more expensive multiplications of arbitrary precision interval values.

These relations between the intervals are lost when Newton steps are performed or the canonical dyadic subdivision points need to be replaced by (pseudo-)admissible points.
Furthermore, optimizations such as the delayed Newton calls are incompatible with strategies that keep only one interval in memory at a time,
because the behavior of the algorithm on some interval is no longer independent of the neighborhood:
whether a quadratic convergence step is pursued becomes conditional on the outcome of the 01-Test on its sibling.

Fortunately, using the Newton technique, long chains of intervals with a unique descendant in the bisection tree are compressed to only logarithmic height, thus shrinking the entire tree by an order of magnitude in the worst-case and even two orders of magnitude for some special instances (e.g. sparse polynomials with clustered roots such as Mignotte polynomials).
In this light, the impact of the constant-memory strategy is much less pronounced, and it is advisable to switch back to the naive approach of caching all intermediate results.
We stress that, without the Newton iteration, the memory consumption in a depth-first traversal would be prohibitive for particular examples with strongly clustered, ill-separated roots such as Mignotte polynomials.

\paragraph{Degree truncation through partial Taylor shifts}
\addcontentsline{toc}{subsubsection}{Degree truncation through partial Taylor shifts}\label{par:partial-TS}
We can consider any univariate polynomial $P(x)$ as a function
$P(x) = p_n \cdot \prod_{j=1}^n (x - z_j)$
in its (not necessarily distinct) complex roots $z_j$.
When the domain of interest is restricted to a small region in the complex plane, say, a disk $D \subset \CC$,
the relative influence of each root $z_j$ on $P$ scales with the distance of $z_j$ to the points $x \in D$.
If there are only $k \ll n$ roots in $D$ and all other roots are well-separated from $D$, the distance to the remote roots is almost stable for all $x \in D$,
and $P \vert_D$ is dominated by the roots in~$D$.

In this situation, the local behavior of $P$ is captured in its partial Taylor expansion around a point within $D$ up to the $k$-th term.
Hence, an approximation of $P$ by its truncated Taylor expansion might suffice for computing isolating intervals for the real roots.
We consider this approach whenever a Newton step succeeds and suggests the existence of a well-separated cluster of $k \ll n$ roots around an interval $I$.
We remark that the multiplicity guess $k \coloneqq \operatorname{round}(\tilde{k})$ is already computed within the Newton-Test.
To ensure correctness, we conservatively take into account the truncation error as the interval evaluation of the Lagrangian remainder term.

More precisely, whenever the inclusion-exclusion-predicates of \ANewDsc call for sign variation tests on some $P_I$,
we work on
an approximation for the intermediate polynomial $Q_I :=\sum_{i=0}^n q_{I,i}x^i:= P(a + (b-a)x)$ on $I = (a,b)$,
\begin{align*}
  \smash{\tilde{Q}}_I(x) &= \sum\nolimits_{i=0}^{k-1} q_{I,i} x^i + [r]_{I,k} x^{k}
  ,\quad [r]_{I,k} \supset \smash{\frac{P^{(k)}(I)}{k!}},
\end{align*}
where the coefficient $[r]_{I,k}$ of the Lagrangian remainder is evaluated on the entire interval $I$ in conservative interval arithmetic.
The root exclusion test on the truncated polynomial remains almost identical; however, we modify the certificate for inclusion of a single root.

\begin{lemma}
  Let $P$ be a polynomial as in \cref{def:polyP} and $I = (a,b)$ and $\smash{\tilde{Q}}_I$ be as above.
  Define $\tilde{G}_I \coloneqq (x+1)^k \;\smash{\tilde{Q}}_I ((x+1)^{-1})$ and $\tilde{H}_I \coloneqq (x+1)^{k-1} \;\smash{\tilde{Q}}'_I ((x+1)^{-1})$.
  \begin{enumerate}[noitemsep]
  \item If $\tilde{v}_I \coloneqq \var(\tilde{G}_I) = \set{0}$, then $P$ has no root in $I$.%
    \footnote{Recall that $\var(\cdot)$, called on an interval polynomial, returns a (super-)set of the possible numbers of sign variations;
      see the remarks about the use of interval arithmetic in RS in \cref{sec:implementation-details}.}
  \item If $\tilde{v}_I = \set{1}$ and, additionally, $\tilde{v}'_I \coloneqq \var (\tilde{H}_I) = \set{0}$, then $P$ has exactly one simple real root in $I$.
  \end{enumerate}
\end{lemma}

\begin{proof}
  We consider the polynomials $\smash{\tilde{Q}}_I$, $\tilde{G}_I$ and $\tilde{H}_I$ simultaneously as polynomials with interval coefficients and as the set of exact polynomials with coefficients contained in the intervals.
  Assume that $\tilde{v}_I = 0$, but $P$ has a root $\xi = a + \lambda (b-a) \in I$.
  By the Lagrange representation of the truncation error in Taylor's theorem, there exists a $\xi\in I$ such that the polynomial
  $
    F_I (x):= \sum\nolimits_{i=0}^{k-1} q_{I,i} x^i+ \frac{P^{(k)}(\xi)}{k!} x^k
  $
  of degree $k$ has a root at $\lambda$.
  Hence, Descartes' Rule of Signs asserts $\var((x+1)^k\,F_I((x+1)^{-1})) > 0$.
  Since $F_I \in \smash{\tilde{Q}}_I$, it follows that $(x+1)^k\,F_I((x+1)^{-1})\in \tilde{G}_I$ because the transformations in interval arithmetic overestimate the possible values. This contradicts the assumption that $\tilde{v}_I={0}$.

  For the second part, an analogous argument on the derivative shows that $P$ is strictly increasing or strictly decreasing on $I$ if $\tilde{v}' = \set{0}$.
  Since $\tilde{v}_I = \set{1}$, we know that $P$ takes values of different sign at the endpoints of $I$; hence, $P$ has exactly one root in $I$.
\end{proof}

If a 01-Test in the processing of a node is unsuccessful, we usually double the working precision.
However, if the local polynomials are only partial Taylor expansions, the source of the loss in accuracy may also stem from the truncation.
Hence, before increasing the working precision, we gradually double the multiplicity guess $k$ until we eventually compute the full polynomial.
This approach guarantees that, in the worst case, after $\log n$ steps the heuristic falls back to the theoretical model.
Hence, even if the heuristic is unsuccessful, the complexity increases at most by a factor of $\log n$.

We observe that partial Taylor shifts give a tremendous speedup for polynomials with tight clusters of low multiplicity.
In such situations, we can consider the technique as a partial substitute for the current deficiency of asymptotically fast Taylor shifts.
For instances with clusters consisting of constantly many roots, the expected speedup is linear in $n$:
if the heuristic applies, the local polynomials on each active region in the subdivision tree have to computed up to only constantly many coefficients of the Taylor series.
This prediction is accurately reflected in the benchmarks.

\section{Experiments}

We present and discuss a short excerpt of our extensive benchmark suite; the entire collection is available online at \url{http://anewdsc.mpi-inf.mpg.de/}.
The objective is to illustrate that the integration of \ANewDsc into RS entails very small overheads for instances with well-distributed roots and shallow subdivision trees,
but huge performance gains in challenging situations with clustered roots.

\flexvspacenoindent

The classical RS serves as a baseline in a configuration similar to the one used in Maple, but without some crucial optimizations to improve comparability.
In particular, we disabled the hardware floating-point phase, which is not yet supported for \ANewDsc, as well as the heuristic for switching to exact arithmetic that can improve performance, but destroys the guarantees on the expected precision demand.
Besides the full-fledged \ANewDsc (denoted as \textsc{AND} in the tables), we also list the intermediate variant \ADsc.
It includes subdivision on pseudo-admissible points and full caching, but neither Newton iterations nor degree truncation.

We compare to three well-established competitors:
\emph{MPSolve} (named MPS in the tables)
is the leading complex root solver.
It is known to keep up with the most efficient real root isolators even though it solves a more general problem.
Given our more modest goals, we restrict the region of interest to the real line, and call MPSolve in its latest version 3.1.5 with
\texttt{unisolve -au -Gi -SR -Dr -Of -j1 -o1048576}.

CF denotes the \emph{continued fraction}-based variant of the Descartes method available in \emph{Mathematica} 10.
We bypass the high-level interface of the computer algebra system via \texttt{System`Private`RealRoots} to eliminate the effect of preprocessing stages that are applied by the usual \texttt{RootIntervals} call
(e.g. the detection and special handling of sparse polynomials or simplifications for even or odd polynomials).

Finally, we compare against Carl Witty's variant of Eigenwillig's bitstream Descartes method in Bernstein basis~\cite{DBLP:conf/casc/EigenwilligKKMSW05},
which constitutes the default real root isolator in the \emph{Sage} 7.0 open-source computer algebra system.
We invoke the algorithm in the optimized variant for 64-bit architectures with \texttt{real\_roots (f, skip\_squarefree=True, wordsize=64)}.

All instances are passed as dense integer polynomials, filtered in a preprocessing stage to factor out side effects of optimizations unrelated to the solver.
It involves the trivial algebraic simplification of even and odd polynomials, verification of square-freeness and reduction to the primitive part.
The values in the tables are runtimes in seconds, measured on a single run on the same machine.
Degree and coefficient bitsize are denoted by $n$ and $\tau$.

\floatplacement{table}{t!}
\begin{table}
  \tablesize
  \begin{tabular*}{\linewidth}{D@{\extracolsep{\fill}}BHTTTTTT}
    \toprule
    $n$      & $\tau$ & \#sol & MPS      & CF       & Sage     & RS       & \textsc{ADsc} & \textsc{AND} \\
    \midrule
    257      & \vcenterwithline{7}{14}     & 3     & 0.7      & 0.1      & 1.6      & 7.6      & 7.7           & 0.1          \\%
    513      &      & 3     & 2.9      & 0.2      & 2.4      & 87.6     & 89.1          & 0.2          \\%
    1025     &      & ERR?  & 13.8     & 1.1      & 4.8      & \timeout & \timeout      & 0.7          \\%
    2049     &      & 3     & 78.5     & 7.2      & 12.8     &          &               & 2.9          \\%
    4097     &      & 3     & 486.3    & 67.6     & 85.8     &          &               & 11.2         \\%
    8193     &      & 3     & \timeout & 224.3    & \timeout &          &               & 43.2         \\%
    16385    &      & 3     & \emph{}  & \timeout &          &          &               & 188.1        \\%
    \cmidrule(lr){1-9}
    \vcenterwithline{6}{129}      & 128    & 3     & 1.2      & 0.3      & 0.5      & 26.2     & 25.1          & 0.1          \\%
          & 512    & 3     & 5.9      & 3.8      & 3.9      & 378.5    & 293.8         & 0.4          \\%
          & 2048   & 3     & 43.4     & 78.4     & \timeout & \timeout & \timeout      & 3.3          \\%
          & 8192   & 3     & 274.7    & \timeout &          &          &               & 20.8         \\%
          & 32768  & 3     & \timeout &          &          &          &               & 96.8         \\%
          & 131072 & 3     & \emph{}  &          &          &          &               & 503.0        \\%
    \bottomrule
  \end{tabular*}
  \caption{Mignotte: $x^n - ((2^{\tau/2}-1)x - 1)^2$}
  \label{table:mignotte}
\end{table}

\begin{table}
  \tablesize
  \renewcommand{\error}{\multirow{4}{*}{\llap{\rotatebox{90}{\parbox{\widthof{\emph{overflow}}}{\centering\emph{call\\stack\\overflow}}}}}}
  \newcommand{\errori}{\multicolumn{1}{c}{\textcolor{badcolor}{\emph{call}}}}
  \newcommand{\errorii}{\multicolumn{1}{c}{\textcolor{badcolor}{\emph{stack}}}}
  \newcommand{\erroriii}{\multicolumn{1}{c}{\textcolor{badcolor}{\emph{over-}}}}
  \newcommand{\erroriv}{\multicolumn{1}{c}{\textcolor{badcolor}{\emph{flow}}}}
  \begin{tabular*}{\linewidth}{D@{\extracolsep{\fill}}BHTTTTTT}
    \toprule
    $n$ & $\tau$ & \#sol & MPS & CF & Sage & RS & \textsc{ADsc} & \textsc{AND}\\
    \midrule
    260 & \vcenterwithline{6}{140} &    12 &   2.7 &   1.0 &   1.9 &   1.6 &   1.7 &   0.4\\%
    516 &  &    12 &   9.5 &  21.6 &   3.3 &  18.0 &  18.0 &   0.8\\%
    1028 &  &    12 &  67.8 &  565.0 & \timeout &  230.3 &  232.4 &   7.1\\%
    2052 &  &    12 &  283.6 &  \timeout &  & \timeout & \timeout &  13.1\\%
    4100 &  &    12 & \timeout &  &  &  &  &  88.4\\%
    8196 &  &    12 & \emph{} &  &  &  &  &  394.0\\%
    \cmidrule(lr){1-9}
    \vcenterwithline{5}{260} & 160 &    12 &   3.1 &   1.4 &   0.3 &   2.2 &   2.3 &   0.9\\%
     & 640 &    12 &   5.8 &  14.9 &  \errori &  20.8 &  19.2 &   0.9\\%
     & 2560 &    12 &  33.7 &  222.6 & \errorii &  301.2 &  254.0 &   3.7\\%
     & 10240 &    12 &  248.0 &  \timeout & \erroriii & \timeout & \timeout &  31.4\\%
     & 40960 &    12 & \timeout &  & \erroriv &  &  &  341.2\\%
    \bottomrule
  \end{tabular*}
  \caption{nested Mignotte: $\prod_{i=1}^4 \big(x^{n/4} - ((2^{\tau/8}-1)x^2-1)^{2i}\big)$}
  \label{table:cascaded-mignotte-irrat}
\end{table}

\begin{table}
  \tablesize
  \begin{tabular*}{\linewidth}{D@{\extracolsep{\fill}}BHTTTTTT}
    \toprule
    $n$ & $\tau$ & \#sol & MPS & CF & Sage & RS & \textsc{ADsc} & \textsc{AND}\\
    \midrule
    1024 & \vcenterwithline{5}{1024} &     4 & {  2.5} &   0.3 &   4.1 &   0.6 &   0.6 &   0.5\\%
    2048 &  &     4 & {  9.8} &   0.9 &   8.4 &   1.6 &   1.7 &   1.7\\%
    4096 &  &     4 & { 37.5} &   2.7 &  29.1 &   3.8 &   3.7 &   3.5\\%
    8192 &  &     4 & { 159.4} &  22.1 &  183.5 &  36.0 &  36.3 &  36.5\\%
    16384 &  &     6 & { 578.9} &  376.6 & \timeout &  275.1 &  280.9 &  279.0\\%
    \bottomrule
  \end{tabular*}
  \caption{dense with uniformly random coefficients in $(-2^\tau,2^\tau) \cap \ZZ$}
  \label{table:random-uniform}
\end{table}
\begin{table}
  \tablesize
  \begin{tabular*}{\linewidth}{D@{\extracolsep{\fill}}BHTTTTTT}
    \toprule
    $n$ & $\tau$ & \#sol & MPS & CF & Sage & RS & \textsc{ADsc} & \textsc{AND}\\
    \midrule
    256 & 506 &    24 &   4.7 &  34.1 &   1.6 &  18.2 &  19.1 &   1.1\\%
    512 & 762 &    24 &  25.9 &  456.9 &   3.7 &  107.7 &  110.4 &   3.0\\%
    1024 & 1274 &    48 &  207.5 & \timeout &  22.5 & \timeout & \timeout &  23.6\\%
    2048 & 2296 &   nan & \timeout &  & \timeout &  &  &  132.2\\%
    4096 & 4343 &    80 & \emph{} &  &  &  &  &  596.6\\%
    \bottomrule
  \end{tabular*}
  \caption{clustered: $f^2 - 1$ for $f = \sum_i a_i \binom{n}{i}^{1/2} x^i / \sqrt{i+1}$ with $a_i$ drawn from a normal Gaussian distribution, rounded to $\ZZ[x]$}
  \label{table:clustered-farahmand}
\end{table}

\flexvspacenoindent
Polynomials of Mignotte type are classical benchmark instances.
Their ill-separated pair of roots at approximately $2^{-\tau/2} \pm 2^{-n\tau/2}$ forces a huge subdivision depth in bisection approaches.
This behavior is mitigated by the quadratic convergence in \ANewDsc:
The subdivision tree size shrinks from approximately 33500 nodes for RS and $\ADsc$ to a mere 47 for $\ANewDsc$ for $(n,\tau) = (129, 512)$,
and even an instance with $\tau=2^{16}$ leads to a tree with only 65 nodes.
For such instances, CF presumably exhibits an almost optimal subdivision tree due to the rational center of the cluster; however, the performance is spoiled by the use of exact arithmetic.
The class of nested Mignotte polynomials, shown around an irrational center in \cref{table:cascaded-mignotte-irrat},
shows the robustness of \ANewDsc both to the higher multiplicity 20 and the irrational center of the cluster.

Random polynomials have a low number ($\Theta(\log n)$) of well-separated real roots; they require only low precision for the isolation and induce flat subdivision trees.
Hence, we can expect no gain from the improvements in \ANewDsc.
On the other hand, the experiments confirm that the enhancements incur almost no additional cost.
We observe similar overheads of low constant factors compared to RS on other classes of polynomials with well-distributed roots,
such as Wilkinson polynomials with evenly spaced real roots, different classes of orthogonal polynomials,
or resultants of random bivariate polynomials that arise in projection-based polynomial system solving and the topology analysis of real algebraic curves.

Finally, polynomials with normal Gaussian-distributed coefficients have, in expectation, significantly more real roots ($\Theta(\sqrt{n})$).
By squaring and perturbing such inputs, we construct benchmark instances with many clusters of multiplicity two.
We find that \ANewDsc quickly detects the clusters, succeeds in the Newton steps, and computes the appropriate number of terms in the partial Taylor shifts.

\flexvspacenoindent
The benchmarks are performed on polynomials with exact integer coefficients.
Nevertheless, we emphasize that \ANewDsc can process inputs with arbitrary approximable, possibly irrational coefficients.
We observe that the performance is unaffected if, for example, the coefficients of the random instances (\cref{table:random-uniform,table:clustered-farahmand}) are drawn from a continuous distribution,
or the coefficients of the Mignotte-like polynomials are irrational numbers of comparable magnitude.

\section*{Acknowledgments}
We thank Maplesoft for their support in the development of RS,
Leonardo Robol (Katholieke Universiteit Leuven) for his help and
support on running the development version of MPSolve
and Adam Strzebonski (Wolfram Research) for his instructions on how to
control the behavior of Mathematica's \texttt{RootIntervals} function.

\clearpage
{%
  \let\oldthebibliography\thebibliography%
  \renewcommand{\thebibliography}[1]{
    \oldthebibliography{#1}
    \addcontentsline{toc}{chapter}{References}
    \setlength{\itemsep}{0pt plus 1ex}
    \raggedright
  }%
  \bibliographystyle{shortabbrv}
  \bibliography{bib}
}

\clearpage
\phantomsection\addcontentsline{toc}{chapter}{Appendix}
\appendix

\section{Comparison to other solvers}

Polynomial root solving is a prevailing and fundamental task in numeric and symbolic computation.
Many different approaches have been successfully implemented with different objectives.
Still, the most successful solvers in widespread use today are variants of only two classes of algorithms:
different approaches based on Descartes' Rule of Signs and the Aberth-Ehrlich iteration for complex root finding.

We compared our implementation of \ANewDsc to the state-of-the-art solvers that we are aware of.
In the following section, we give a brief survey of the characteristics, weaknesses and strengths of the main candidates for comparison.
We do so without any claim for completeness; it is impossible to give due credit to all capabilities of those mature solvers.
Also, we ignore methods based on Sturm sequences or Pan's asymptotically near-optimal polynomial factorization, as we are not aware of competitive implementations of those methods.

\paragraph{MPSolve}
\addcontentsline{toc}{subsection}{MPSolve}\label{par:comp-mpsolve}
As the only complex root solver, MPSolve stands out in this list.
It is based on the Aberth-Ehrlich root finding iteration, a modified Newton method that simultaneously refines approximations for all complex roots.
This numerical method is very efficient in practice, despite the fact that it lacks theoretical convergence guarantees.
The computations are done in pure multiprecision arithmetic without interval arithmetic, but the results are certified based on a priori round-off error bounds an a posteriori Gershgorin-type argument.

A major selling point of MPSolve is its triple-stage approach: first, the root isolation is tried with pure hardware floating-point numbers.
If this is unsuccessful, a \texttt{double} type with extended exponent range is employed, and only if absolutely necessary, multiprecision arithmetic is used.
Hence, MPSolve excels for instances where the roots can be isolated within low precision.
In such situations, it is on par with dedicated real solvers despite answering a more general problem.
Since version 3, another still unique feature of the solver is its full support for multithreading.

Unfortunately, until this date the Aberth-Ehrlich iteration comes without any termination or even worst-case complexity guarantees.
In constrast to all other methods in the comparison, the method is intrinsically global;
solutions or clusters of roots can be excluded in the later stages of the algorithm, but the boundaries of the region of interest is fuzzy, and nearby roots affect the behavior of the algorithm.
Finally, the certification method is extremely sensitive to a very careful theoretical analysis and an accurately matching implementation.
For example, once fast polynomial becomes beneficial for real-world instances, the necessary modifications of the validation routines require a new explicit analysis of the round-off error bounds.%
\footnote{Private communication with Leonardo Robol.}

For our benchmarks, we use MPSolve in the most recent development version 3.1.5.
MPSolve 2.2 is still very slightly superior on selected instances, but both for MPSolve 2.2 and the most recent secsolve algorithm (which used a dedicated solver tailored for secular equations), we experienced minor bugs in the certification phase.
The authors are currently investigating the problems and already corrected some of the issues.%
\footnote{While the bugs in secsolve have been resolved, it will be difficult to backport the necessary fixes for MPSolve 2.2 (private communication with Leonardo Robol).}
For the time being, we only run the unisolve configuration of MPSolve 3 using the call
\texttt{unisolve -au -Gi -SR -Dr -Of -j1 -o1048576} on dense inputs with exact integer coefficients.
We remark that the \texttt{-SR} option restricts the region of interest to the real line, in agreement with our more modest goals.

\paragraph{Continued fraction-based solver in Mathematica}
\addcontentsline{toc}{subsection}{Continued fraction-based solver in Mathematica}\label{par:comp-cf}
The implementation in Mathematica
is a highly efficient pure real root solver.
As RS and the other real solvers, it relies on Descartes' Rule of Signs, but uses a different subdivision strategy:
successive computations of root bounds for local polynomials are exploited to compute continued fraction approximations of the roots.

The major benefit of this approach is an extremely fast convergence to rational roots and the choice rational subdivision points with low bitsize of numerator and denominator.
For polynomials with mostly rational and well-separated roots, this method is often superior to other competitors.

However, the price paid in the present variant of the solver is that exact arithmetic is used.
Hence, there is no easy way to adapt for bitstream inputs where it can be impossible to decide conclusively whether a subdivision point is a root or not,
and termination is not clear without specific safety measures if a root is chosen.
Furthermore, even if the algorithm can quickly solve an instance $f \in \ZZ[x]$, the algorithm can be forced into extremely expensive exact computations on inputs with high bitsize,
say $c\cdot f - 1$ for some huge $c \in \ZZ$, despite the fact that the intrinsic difficulty as a numerical problem remains the same.

The continued fraction scheme achieves optimal convergence against rational roots.
However, irrational algebraic numbers are not necessarily easy to approximate in this setting: the golden ratio $\Phi = (1 + \sqrt{5})/2$ is a notoriously hard example and shows the slowest convergence rate possible.
Thus, it is easy to construct instances that force continued fraction-based solvers into linear convergence.
For example, the Mignotte-like polynomial $x^n - (a\,x-1)^2$ for an integer $a \in \ZZ_{>1}$ has roots at approximately $1/a \pm a^{-n/2}$.
A subdivision at exactly $1/a$ is easily achieved by the method and a perfect split for the cluster of two roots.
However, this is merely an artifact by construction: the polynomial $x^n - (a\,x^2-1)^2$ for $a \in \ZZ_{>1}$ not a square number has roots at approximately $1/\sqrt{a} \pm a^{-n/4}$.
But until the algorithm arrives at a sufficiently good rational approximation of the center of the cluster to split the roots, it has to go through through $\Omega (n)$ iterations with only linear convergence.

Those shortcomings are responsible for a loss of one and two orders of magnitudes in the worst-case complexity, and can be easily be enforced by an adversary.

For a fair comparison, we call the continued fraction solver in Mathematica 10 with its default configuration directly via \texttt{System`Private`RealRoots} to bypass the high-level interface.
The latter performs several preprocessing operations, such as detection and optimized handling of sparse polynomials or the replacement of even polynomials $f(x^2)$ by $f(x)$ and reconstruction of the original roots later on.

\paragraph{Sage}
\addcontentsline{toc}{subsection}{Sage}\label{par:comp-sage}
The computer algebra system Sage
collects a wealth of high-quality packages and libraries from the open source world in one easily accessible framework, complemented with original implementations of new algorithms.
It includes a real solver written in Cython by Carl Witty, which is loosely based on the Bitstream Descartes method by Eigenwillig~\cite{eigenwillig-phd}.
The routine by the apt name \texttt{real\_roots} is unique as a very efficient root finder written in a high-level language, but relies on the well-established MPFI library for interval arithmetic in arbitrary precision.
All computations are performed in Bernstein basis using de Casteljau's algorithm, and heuristics for degree reduction are used to achieve similar effects as the partial Taylor shifts in \ANewDsc.
With respect to the convergence speed, the implementation notes state that a hardware-oriented strategy is used for the selection of the subdivision points.
To the best of our knowledge, however, there is no thorough description or analysis of worst- or expected-case complexity of the algorithm.

Sage's \texttt{real\_roots} is a very good general-purpose solver.
Due to the extensive use of approximate arithmetic, it stays close to the optimum precision demand and is conceptually suited to be used as a bitstream solver.
In our experiments, \texttt{real\_roots} converged quickly to tight clusters.
Nevertheless, such situations seem to be its Achilles' heel:
In contrast to the competitors, \texttt{real\_roots} is implemented in a recursive rather than iterative fashion.
For inputs that force a very deep computation tree, the algorithm is prone to call stack overflows when there are deeply nested clusters and no tail call elimination applies.

We call the solver in release 7.0 of Sage, using the default settings without square-free decomposition as a preprocessing step, via
\texttt{real\_roots (f, skip\_squarefree=True, wordsize=64).}

\paragraph{Classical RS in Maple}
\addcontentsline{toc}{subsection}{Classical RS in Maple}\label{par:classical-rs}
The classical RS version has already been described before.
It serves as a general-purpose solver in the Maple computer algebra system.
Over the course of fifteen years of development, a collection of many optimizations and heuristics have been integrated to suit the requirements both as an internal routine in Maple and for direct users.

Some of these improvements are incompatible or significantly less effective when combined with \ANewDsc.
Others sacrifice the theoretical optimality or are not suited for bitstream inputs; most prominently, this applies to the heuristics for switching to exact arithmetic.
Despite the fact that exact computations are only selectively used, the classical RS version can be forced into asymptotically bad behavior with high bitsize inputs in a similar manner as the solver in Mathematica.
Hence, the baseline for our experiments is not the version inside Maple, but rather a simplified variant that focuses on simplicity and provably asymptotic optimality.
In particular, so far none of the new routines for \ANewDsc are available for interval arithmetic with hardware floating point numbers.

The loss in performance of the stripped-down version is often negligible, but can raise up to a factor of 10 for very particular instances with low intrinsic difficulty such as Wilkinson polynomials,
where all roots are real and well-structured, exact arithmetic is asymptotically optimal with respect to the precision demand, and optimizations apply to greatest extent.

\section{Benchmark suite}

All our measurements are performed on a server with four Intel Xeon E5-2680 v3 processors with twelve cores each, clocked at 2.5 GHz.
The CPU has 30 MB shared L3 cache, and each core has 256 KB L2 cache.
The server has a total of 256 GB RAM and runs on Debian 8 ``jessie'' 64-bit.
The measurements are taken from a single run with a timeout of 10 minutes; for timings above 0.1 seconds, we experienced a very stable timing with deviations between runs below one percent for all solvers.
However, since the running times are not averaged over many runs, the results on random inputs are only comparable horizontally, not vertically:
the solvers are called on the same polynomials, but only one instance for each magnitude is considered.
All polynomials are passed as dense polynomials with arbitrary precision integer coefficients.
The parsing time for the input is not factored out:
we noticed that the parsing in our preliminary interface for \ANewDsc is very slow, and for random high-bitsize inputs can take significantly longer than the total time to completion for other solvers.
For the time being, the instances are of moderate size such that reading the input is almost negligible even for the variants of RS.

Our full benchmark suite as well as our software for popular architectures (Linux, OSX and Windows) is available online at  \url{http://anewdsc.mpi-inf.mpg.de/}.
The readers are invited to verify our results, include their favorite alternative solvers, and inform us about meaningful extensions.
We hope that the collection of polynomials grows to a representative corpus of inputs to put polynomial root isolators to the acid test,
and look forward to suggestions of both artificial and realistic instances.

\end{document}